\documentclass[envcountsame]{llncs} 
\usepackage{amsmath,amssymb,amsfonts,latexsym,epsfig,graphicx} 
\usepackage{psfrag,xspace}
\usepackage{multirow} 
\usepackage{epic}
\usepackage{pgf} 
\usepackage{subfigure}
\def\s{{\bigtriangleup}}
\def\CC{{\mathbb C}} \def\NN{{\mathbb N}} \def\RR{{\mathbb R}} 

\title{Algebraic methods for counting Euclidean embeddings of rigid graphs}

\author{Ioannis Z.\ Emiris \inst{1} \and  Elias P.\ Tsigaridas\inst{2} \and
  Antonios Varvitsiotis\inst{1} } 

\institute{
  National and Kapodistrian University of Athens, Greece. \and INRIA     M\'editerran\'ee, Sophia-Antipolis, France.}

\begin{document}
\maketitle

\begin{abstract}
  The study of (minimally) rigid graphs is motivated by numerous applications,
  mostly in robotics and bioinformatics. 
  A major open problem concerns the number of embeddings of such graphs,
  up to rigid motions, in Euclidean space. We capture embeddability by polynomial systems
  with suitable structure, so that their mixed volume, which bounds the number 
  of common roots, to yield interesting upper bounds on the number of embeddings.
  We focus on $\RR^2$ and $\RR^3$, where Laman graphs and 1-skeleta of convex
simplicial polyhedra, respectively, admit inductive Henneberg
constructions. We establish the first  lower bound in $\RR^3$ of about
$2.52^n$, where $n$ denotes the number of vertices. Moreover, our implementation
yields upper bounds for $n \le 10$ in $\RR^2$ and $\RR^3$, which reduce the existing gaps, and tight bounds  up to  $n=7$ in $\RR^3$.

\medskip

\end{abstract} 

\section{Introduction}

Rigid graphs (or frameworks)
constitute an old but still very active area of research due to their deep
mathematical and algorithmic questions, as well as numerous
applications, e.g.\ mechanism and linkage theory \cite{WH07,WH07b},
and structural bioinformatics \cite{EM99,JRKT01,TD99}. 

Given a graph $G=(V,E)$ and a collection of edge lengths
$l_{ij}\in\RR^{+}$, for $(i,j) \in E$, a {\em Euclidean embedding} in $\mathbb{R}^d$
is a mapping of $V$ to a set of points in $\mathbb{R}^d$, such that $l_{ij}$ equals
the Euclidean distance between the images of the $i$-th and $j$-th vertices,
for $(i,j) \in E$.
Euclidean embeddings impose no requirements on whether the edges cross each other or not.
A graph is (generically) {\em rigid} in $\mathbb{R}^d$ iff, for generic (or random)
edge lengths, it is embedded in $\RR^d$ in a finite number of ways, modulo rigid motions.
A graph is {\em minimally rigid} iff it is no longer rigid once any edge is removed.

A graph is called {\em Laman} iff $|E|=2|V|-3$ and, additionally, 
all of its induced subgraphs on $k<|V|$ vertices have $\le 2k-3$ edges.
The class of Laman graphs coincides with  the generically minimally rigid graphs in $\mathbb{R}^2$,
and also admit inductive constructions.
In $\RR^3$ there is no analogous combinatorial characterization
of generically rigid graphs. On the other hand, the 1-skeleta, or edge graphs,
of (convex) simplicial polyhedra are minimally rigid in $\mathbb{R}^3$,
and admit inductive constructions, cf.\ Section\ \ref{Sgenerate3d}.

In this paper, we deal with the problem of computing the maximum number
of distinct planar and spatial Euclidean embeddings of (minimally) rigid graphs,
up to rigid motions, as a function of the number of vertices.
To study upper bounds, we define a square polynomial system, 
expressing the edge length constraints, whose real solutions correspond
precisely to the different embeddings. 
Here is a system expressing embeddability in $\RR^3$, where $(x_i,y_i,z_i)$ are the coordinates
of the $i$-th vertex, and 3 vertices are fixed to discard translations and rotations:
\begin{equation}\label{Esystem} \small \left\{ \begin{array}{ll}
      x_i=a_i, \; y_i=b_i, \; z_i=c_i,&  i=1,2,3 , \quad a_i,b_i,c_i\in\RR, \\
      (x_i-x_j)^2+(y_i-y_j)^2+(z_i-z_j)^2 = l_{ij}^2,\;\; & (i,j)\in E-\{(1,2),(1,3),(2,3)\}
    \end{array} \right.
\end{equation}
All nontrivial equations are quadratic; there are $2n-4$ for Laman graphs, and $3n-9$
for 1-skeleta of simplicial polyhedra, where $n$ is the number of vertices.
The classical B\'ezout bound on the number of roots equals the product of the
polynomials' degrees, and yields $4^{n-2}$ and $8^{n-3}$, respectively.
It is indicative of the hardness of the problem that efforts to substantially
improve these bounds have failed.

For the planar and spatial case, the best known upper bounds are 
$\binom{2n-4}{n-2} \approx 4^{n-2}/ \sqrt{\pi (n-2)}$ and
$\frac{2^{n-3}}{n-2}\binom{2n-6}{n-3} \approx 8^{n-3}/\big((n-2)\sqrt{\pi (n-3)}\big)$, respectively.
These bounds were obtained using 
complex algebraic geometry \cite{B02,BS04}.
For $\RR^2$, there exist   lower bounds of  
$24^{\lfloor (n-2)/4 \rfloor}\simeq 2.21^n$ and $2\cdot 12^{\lfloor (n-3)/3 \rfloor}\simeq 2.29^n/6$, obtained  by  a caterpillar and a 
fan\footnote{This corrects the exponent of the original statement.} construction, respectively. 
Both of these bounds  are based on the  Desargues (or 3-prism) graph
(Figure~\ref{desargues graph}).

In applications, it is crucial to know the number of embeddings for specific 
(small) values of $n$.
The most important result in this direction was to show that the Desargues graph  admits 24
embeddings in the plane~\cite{BS04}. Moreover, the $K_{3,3}$ graph admits 16 embeddings in the plane~\cite{W77,WH07}  and 
the cyclohexane graph  admits 16 embeddings in space~\cite{EM99}.


Mixed volume (or Bernstein's bound) of a square polynomial system 
exploits the sparseness of the equations to bound the number of common roots,
it is always bounded by B\'ezout's bound and typically much tighter, cf. Section~\ref{sec:abid}.
We have implemented specialized software that  constructs  all rigid graphs up to isomorphism,
for small $n$, and computes  the mixed volumes of their respective polynomial systems.
Our results indicate that mixed volume can be of general interest in enumeration problems.

Our main contribution is twofold, besides some straightforward
upper bounds in Lemmas~\ref{Lsparse2} and~\ref{Lsparse3}. 
%
First, we derive the first lower bound in $\RR^3$:
$$
16^{\lfloor (n-3)/3\rfloor} \simeq 2.52^n,\, n\ge 9,
$$
by designing a cyclohexane caterpillar.

Second, we obtain upper and lower  bounds for $n \le 10 \ $ in  $\RR^2$  and $\RR^3$, which  reduce the existing gaps, 
see Tables~\ref{tab:2D-constructions} and \ref{tab:3D-constructions} in the Appendix.
Moreover, we establish tight bounds up to $n = 7$ in $\RR^3$ by
appropriately formulating the polynomial system.
We apply Bernstein's Second theorem to show that the naive polynomial
system cannot yield tight mixed volumes,
a fact already observed in the planar case~\cite{ST08}.

The rest of the paper is structured as follows: 
Section~\ref{sec:laman-2d} discusses the case $d=2$,
Section~\ref{sec:abid} presents our algebraic tools and our implementation,
Section~\ref{Sgenerate3d} deals with $\RR^3$, and we conclude with open questions
and a conjecture.
Omitted Figures and Tables have been  included  in the Appendix.

Some results appeared in \cite{EV09} in preliminary form.

\section{Planar embeddings of Laman graphs} \label{sec:laman-2d} 

 
 Laman graphs admit inductive constructions  that begin with  a triangle, followed by a sequence of Henneberg-1 (or $H_1$) and Hennenerg-2 steps (or $H_2$).  Each such step adds a new  vertex as follows:
a $H_1$ step  connects it to two existing vertices;
a $H_2$ step  connects it to three existing vertices having
at least one edge among them, which is removed (Figure~\ref{planar Henneberg}).
We represent each Laman graph by $\bigtriangleup s_4 \ldots,s_n$, where $s_i \in\{1,2\}$; this is known as
 its {\em Henneberg sequence}.
A Laman graph  is called  $H_1$ iff it can be constructed using only $H_1$ steps, and
 $H_2$ otherwise.
Since two circles intersect generically in two points, a $H_1$ step  at most doubles
the number of embeddings and this is tight, generically. It follows that a $H_1$  graph on $n$ vertices has  $2^{n-2}$ embeddings~\cite{BS04}.

One can easily verify that every $\bigtriangleup 2$ graph is isomorphic to a $\bigtriangleup 1$ graph and that 
every $\bigtriangleup 12 $ graph is isomorphic to a $\bigtriangleup 11$ graph. Consequently, all Laman graphs with  $n=4,5$ are $H_1$ and they have 4 and 8 embeddings, respectively.   

For a Laman graph on 6 vertices, there are 3 possibilities: it is either $H_1$,   $K_{3,3}$ or the Desargues graph. Since the $K_{3,3}$ graph  has at most 16 embeddings~\cite{W77,WH07} and  the Desargues graph has 24 embeddings~\cite{BS04}, the uppper bound is 24  for $n=6$.


\begin{lemma}\label{laman64}
The maximum number of Euclidean embeddings for Laman graphs with $n= 7,\dots,10$ 
is 64, 128, 512 and 2048, respectively.
\end{lemma}

\begin{proof}
Using our software (Section~\ref{sec:abid}), we construct all Laman graphs with $n=7,\dots,10$,
and compute their  respective mixed volumes, thus obtaining the upper bounds.
\end{proof}

Table~\ref{tab:2D-constructions} in the Appendix summarizes our results for $n\le 10$.
The lower bound for $n= 9$ 
follows from  the Desargues fan \cite{BS04}. All other lower bounds follow  from   the fact that a $H_1$ step exactly doubles the number of embeddings.

We now establish an  upper bound, which improves upon the existing ones when our 
graph  contains many degree-2 vertices.
Our proof parallels that of \cite{ST08} which uses mixed volumes to bound the effect
of a $H_1$ step, when it is the last one in the Henneberg sequence.

\begin{lemma}\label{Lsparse2}
  Let $G$ be a Laman graph with  $k \ge 4$ degree-2 vertices.
  Then, the number of planar embeddings of $G$ is bounded above by $2^{k-4}     4^{n-k}$.
\end{lemma}

\begin{proof}
The removal of the existing $k$ degree-2 vertices  does not affect any of the existing degree-2 vertices (because the
  remaining graph  is also Laman), although it may create new ones.
  Since the remaining graph has $n-k$ vertices, the B\'ezout bound of its
polynomial system is equal to $4^{n-k}$ and thus  the number of
embeddings is bounded above by    $2^{k-4} 4^{n-k}$.
\end{proof}

\section{An algebraic interlude} \label{sec:abid}

This section introduces mixed volumes and discusses our computer-assisted proofs;
for background see \cite{B75,EmCa95} and references therein.

Given a polynomial $f$ in $n$ variables, its support is the set of exponents
in $\NN^n$ corresponding to nonzero terms (or monomials).
The Newton polytope of $f$ is the convex hull of its support and lies in $\RR^n$.
Consider polytopes $P_i\subset\RR^n$ and parameters $\lambda_i\in\RR, \lambda_i\ge 0$, for $i=1,\dots,n$.
%
Consider the Minkowski sum of the scaled polytopes
$\lambda_1 P_1+\cdots+\lambda_nP_n \in\RR^n$;
its (Euclidean) volume is a homogeneous polynomial of degree $n$ in the $\lambda_i$.
The coefficient of $\lambda_1\cdots \lambda_n$ is the
{\em mixed volume} of $P_1,\ldots,P_n$.
If $P_1 = \cdots =P_n$, then the mixed volume is $n!$ times the volume of $P_1$.
We focus on the topological {\em torus} $\CC^*=\CC-\{0\}$.
 
\begin{theorem} {\rm \cite{B75}}
Let $f_1=\cdots= f_n=0$ be a polynomial system in $n$ variables with real coefficients,
where the $f_i$ have fixed supports.
  The number of isolated common solutions in $(\mathbb{C}^*)^n$ is bounded above 
  by the {\em mixed volume} of (the Newton polytopes of) the $f_i$.
  This bound is tight for a generic choice of coefficients of the $f_i$'s.
\end{theorem}

Bernstein's Second Theorem below, describes genericity.
Given $v \in \RR^n-\{0\}$ and polynomial $f_i$, we denote by $\partial_v f_i$
the polynomial obtained by keeping only those terms whose exponents minimize
the inner product with $v$.
The Newton polytope of $\partial_v f_i$ is the face of the Newton polytope of $f_i$
supported by $v$.

\begin{theorem}\label{Bernstein2nd} {\rm \cite{B75}}
  If for all $v \in \RR^n-\{0\}$ the face system
  $\partial_vf_1=\ldots=\partial_vf_n=0$ has no solutions in $(\mathbb{C^*})^n$,
then the mixed volume of the $f_i$ exactly equals the number of
solutions in $(\mathbb{C^*})^n$, and all solutions are isolated.
Otherwise, the mixed volume is a strict upper bound on the number of isolated solutions.
\end{theorem}

This theorem was used to study planar embeddings \cite{ST08}; we shall apply it to $\RR^3$.

Now, we describe how we make use of mixed volume in our implementations.
In order to bound the number of embeddings of rigid graphs, we have
developed specialized software that constructs all Laman graphs 
and all 1-skeleta of simplicial polyhedra in $\RR^3$ with $n\le 10$.
Our computational platform is SAGE\footnote{\texttt{http://www.sagemath.org/}}.
We exploit the fact that these graphs admit inductive constructions, and
construct all graphs using the  Henneberg steps.
The latter were implemented, using SAGE's interpreter, in Python. 
After we construct all the graphs, we classify them up to isomorhism 
using SAGE's interface for N.I.C.E., an open-source isomorphism check engine,
keeping for each graph the Henneberg sequence with largest number of $H_1$ steps.
 
Then, for each graph we construct a polynomial system whose real solutions express all possible embeddings,
using formulation~(\ref{Enewsys}) below.
For each system we bound the number of its (complex) solutions by computing its
mixed volume, using \cite{EmCa95}.
Notice that, by genericity, solutions have no zero coordinates.
%
For every Laman graph, to discard translations and rotations,
we assume that one edge is of unit length, 
aligned with an axis, with one of its vertices at the origin.
In $\RR^3$, a third vertex is also fixed so as to belong to a coordinate plane.
The corresponding coordinates are given specific values and are no longer unknowns.
Depending on the choice of the fixed edge, we obtain different systems hence different
mixed volumes. Since they all bound the actual number of embeddings,
we use the minimum mixed volume.  

We used an Intel Core2, at 2.4GHz, with 2GB of RAM.
We tested more that $20, 000$ graphs and computed the mixed 
volume of more than $40, 000$ polynomial systems.
The total time of experiments was about 2 days.
Tables~\ref{tab:2D-constructions} and \ref{tab:3D-constructions} summarize our results. 

\section{Spatial embeddings of 1-skeleta of simplicial polyhedra} \label{Sgenerate3d}

This section extends the previous results to 1-skeleta of (convex) simplicial polyhedra,
which are minimally rigid in $\RR^3$ \cite{G75}.
For such a graph $(V,E)$, we have $|E|=3|V|-6$ and
all of the induced subgraphs on $k<|V|$ vertices have $\le 3k-6$ edges.

Consider any $k+2$ vertices forming a cycle with $\ge k-1$ diagonals, $k\ge 1$.
The extended Henneberg-$k$ step (or $H_k$), $k=1,2,3$, corresponds to
adding a vertex, connecting it to the $k+2$ vertices, and
removing $k-1$ diagonals among them, cf.\ Figure~\ref{hennsteps}.
%
It is known that a graph is the 1-skeleton of a simplicial polyhedron in $\RR^3$ iff
it has a construction that begins with the 3-simplex followed by any sequence
of $H_1, H_2, H_3$ steps \cite{BF67}.

Since 3 spheres intersect generically in two points, a $H_1$ step 
at most doubles the number of spatial embeddings and this is tight, generically.
%
In order to  discard translations and rotations, we fix a (triangular)
facet of the polytope; we choose wlog the first 3 vertices and obtain
system~(\ref{Esystem}) of dimension $3n$.
It turns out that this system does not capture the structure of the problem.

Specifically, choosing direction  $v=(0,0,0,0,0,0,0,0,0,-1,\ldots,-1)\in \mathbb{R}^{3n}$, the corresponding  face system is:
\begin{equation*} \small \left\{ \begin{array}{ll}
      x_i=a_i, y_i=b_i, z_i=c_i,&i=1,2,3 , \quad a_i,b_i,c_i\in\RR , \\
      (x_i-x_j)^2+(y_i-y_j)^2+(z_i-z_j)^2=0, \;\; & (i,j) \in E,\ i,j \not\in\{ 1,2,3 \},\\
      x_i^2+y_i^2+z_i^2=0,& (i,j) \in E:\ i \notin\{1,2,3\},\ j \in \{1,2,3\}.
    \end{array}\right.  \label{system}
\end{equation*}
This system has
$(a_1,b_1,c_1,\ldots, a_3,b_3,c_3, 1,1,\gamma\sqrt{2},\ldots, 1,1,\gamma\sqrt{2})$
$\in(\CC^*)^{3n}$ as a solution, where $\gamma=\pm\sqrt{-1}$.
According to Theorem~\ref{Bernstein2nd}, the mixed volume is not
a tight bound on the number of solutions in $(\CC^*)^{3n}$.
This was also observed, for the planar case, in \cite{ST08}.
To remove spurious solutions (at toric infinity),
we introduce variables $w_i=x_i^2+y_i^2+z_i^2$, for $i=1,\ldots,n$.
This yields the following equivalent system, but with lower mixed volume:
\begin{equation}\label{Enewsys} \small \left\{ \begin{array}{ll}
    x_i=a_i, y_i=b_i, z_i=c_i,&  i=1,2,3 , \\
    w_i = x_i^2 + y_i^2 + z_i^2,&  i=1,\ldots,n, \\
    w_i+w_j-2x_ix_j-2y_iy_j-2z_iz_j = l_{ij}^2, \;\; & (i,j)\in E-\{(1,2),(1,3),(2,3)\}.
  \end{array} \right.
\end{equation}
This is the formulation we will use for our computations.
 
For $n=4$, the only simplicial polytope is the 3-simplex, which 
admits 2 embeddings.  
For $n=5$, there is a unique graph that corresponds to a  1-skeleton of a simplicial polyhedron, see Figure~\ref{spatial5}
\cite{BF67}.
This graph is obtained from the 3-simplex trough a $H_1$ step, so for $n=4$ the bound is 4 and it is tight.

\begin{theorem}\label{16embeddings}
  The 1-skeleton of a simplicial polyhedron on 6 vertices has at most 16     embeddings
  and this bound is tight.
\end{theorem}
  
  \begin{proof}
  There are two non-isomorphic graphs $G_1,G_2$ for $n=6$ \cite{BF67},
  see Figure~\ref{fig:1-skeleta-v6}.
For $G_1$, the mixed volume is 8.
Since all facets of $G_2$ are symmetric, we fix one and compute the
mixed volume, which equals 16, so the upper bound is 16.
$G_2$ is the graph of the cyclohexane, which admits 16 different 
Euclidean embeddings \cite{EM99}.
To see the equivalence, recall that the cyclohexane is essentially a 6-cycle
(fig. \ref{cyclohexane configuration})
with known lengths between vertices at distance 1 (adjacent) and 2.
The former are bond lengths whereas the latter are specified
by the angle between consecutive bonds.\end{proof}

We now pass to general $n$ and establish the first  lower bound in $\RR^3$.

\begin{theorem}\label{Tcyclocaterp}
There exist edge lengths for which the cyclohexane caterpillar
  construction has $16^{\lfloor (n-3)/3 \rfloor} \simeq 2.52^n$ embeddings, for $ n\ge 9$.
\end{theorem}
\begin{proof}
  We glue together copies of cyclohexanes sharing a common triangle.   The resulting graph is the 1-skeleton of a simplicial polytope. Each new copy adds 3 vertices,
and since there exist edge lengths for which the cyclohexane graph has  16 embeddings the claim follows.  An example with 2 copies is illustrated in Figure~\ref{fig:cyclohexane}.

\end{proof}

Table~\ref{tab:3D-constructions} summarizes our results for $n \le 10$.
The upper bounds for $n=7,\dots,10$ are computed by our software.
The lower bound for $n=9$ follows from the cyclohexane caterpillar.
All  other lower bounds are 
obtained by applying a $H_1$ step to a graph with one fewer vertex.
Lastly, we state without proof  a result similar to Lemma~\ref{Lsparse2}. 

\begin{lemma}\label{Lsparse3}
  Let $G$ be the 1-skeleton of a simplicial polyhedron with $k \ge 9$ 
  degree-3 vertices.
  Then the number of embeddings of $G$ is bounded above by $2^{k-9}  8^{n-k}$.
\end{lemma}

\section{Further work}

Undoubtedly, the most important and oldest problem in rigidity theory
is the combinatorial characterization of rigid graphs in $\mathbb{R}^3$.
In the planar case, existing bounds are not tight despite our earlier
attempt to close all gaps up to $n=8$. This is due to the fact that
root counts include rotated copies of certain embeddings. We expect
that a more careful approach should settle these cases.
Since we deal with Henneberg constructions,
it is important to determine the effect of each step on the number of embeddings:
a $H_1$ step always doubles their number;
we {\em conjecture} that $H_2$ multiplies it by $\le 4$ and spatial $H_3$ by $\le 8$,
but these may not always be tight.
Our conjecture has been verified for small $n$.
Here, mixed volume may help; it is also challenging to understand
its relevance, despite the fact it is a bound on complex roots.
As for lower bounds, specific graphs with a large number of
embeddings may be combined to yield tighter bounds.

\medskip

\noindent
\textbf{Acknowledgement.}
A.~V.\ thanks G\"unter Rote for insightful discussions on our conjecture. 
A.~V. acknowledges partial support 
by IST Programme of the EU as a Shared-cost RTD (FET Open) Project under Contract 
No IST-006413-2 (ACS - Algorithms for Complex Shapes).
I.~E.\ is partially supported by FP7 contract PITN-GA-2008-214584 "SAGA''; part of this work was done while he was on sabbatical
leave at "Ecole Normale de Paris" and "Ecole Centrale de Paris".
E.~T.\ is partially supported by contract ANR-06-BLAN-0074 "Decotes''. All authors thank D.~Walter for insightful discussions 
on the number of embeddings of the $K_{3,3}$.
\bibliographystyle{plain}
\bibliography{emiris,rigidity,diplomatiki}


\section*{Appendix}

\begin{table}[h]
  \centering  $ \begin{array}{c||c|c|c|c||c|c|c|c}
    n = & 3& 4 & 5 & 6 & 7 & 8& 9 & 10\\
    \hline \hline
    \mbox{lower} & 2& 4 & 8 & 24 & 48 & 96& 288 & 576 \\
    \hline \hline
    \mbox{upper}& 2& 4 & 8 & 24 & 64 & 128& 512 & 2048\\
    \hline \hline
    H_1 & \s &\s1 & \s11& \s 111 & \s1^4 & \s1^5& \s1^6&\\
    \hline
    &         &       &          &\mathbf{ \s112} & \mathbf{\s1^32} & \mathbf{\s1^42} & \s1^52 &\\
    H_2  &&&&&&& \s1^421&\\
    &&&&&&& \mathbf{\s1^422} &\\
    \hline
  \end{array} $ 
\medskip
  \caption{Bounds and Henneberg sequences for Laman graphs for $n\le 10$. 
Bold text indicates the Henneberg sequence of the graph yielding the upper bound.}
    \label{tab:2D-constructions}
\end{table}

\begin{table}[h]
  \centering  $ 
\begin{array}{c||c|c||c|c|c|c|c}
      n = & 4 & 5 &6 &7 & 8 & 9  &10\\
      \hline
      \hline
      \mbox{lower}&2&4 & 16 & 32 & 64 & 256 & 512 \\
      \hline
      \hline
      \mbox{upper} &2& 4& 16 & 32 & 160 & 640 & 2560  \\
      \hline
      \hline
      H_1 & \s & \s 1& \s 11& \s111& \s 1^4 &\s 1^5 & \s 1^6\\
      \hline
      &&& \mathbf{\s 12} &\mathbf{ \s1^22} & \s1^32&\s 1^42  &\s1^52\\
      &&&&& \mathbf{\s 1^22^2} & \s 1^3 2^2&\s1^42^2\\
      H_2 &&&&&\s 1^221&\s1^321&\s1^421\\
      &&&&&& \mathbf{\s 1^22^3}&\s1^32^3\\
      &&&&& &&\s1^321^2\\
      &&&&&&&\s1^3212\\
      &&&&&&&\s1^32^21\\
      &&&&&&& \mathbf{\s1^22^4}\\
      \hline
      H_3 &&&&&&&\\
      \hline
    \end{array} $ 
\medskip
  \caption{Bounds and Henneberg sequences for 1-skeleta of simplicial polyhedra for
$n\le 10$, where $\s$ is the 3-simplex.  $H_3$ need not apply before $n=13$.
Bold text indicates the Henneberg sequence of the graph that yields  the upper bound.}
  \label{tab:3D-constructions}
\end{table} 

\begin{figure}[h]
  \begin{minipage}[b]{0.5\linewidth} 
    \centering
    \includegraphics{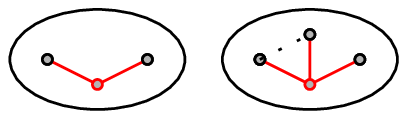}
    \caption{The planar Henneberg steps; the bottom vertex is new.}
    \label{planar Henneberg}    
  \end{minipage}
  \hspace{0.1cm} 
  \begin{minipage}[b]{0.5\linewidth}
    \centering
    \includegraphics{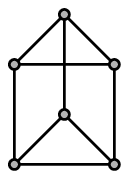}
    \caption{The Desargues graph.}
  \label{desargues graph}    
  \end{minipage}
\end{figure}

\begin{figure}[h]
    \centering
    \begin{pgfpicture}
        \color[rgb]{0,0,0}\pgfsetlinewidth{0.254mm}\pgfline{\pgfpoint{20.32mm}{-10.16mm}}{\pgfpoint{6.35mm}{-24.13mm}} 
        \color[rgb]{0,0,0}\pgfsetlinewidth{0.254mm}\pgfline{\pgfpoint{20.32mm}{-10.16mm}}{\pgfpoint{34.29mm}{-24.13mm}} 
        \color[rgb]{0,0,0}\pgfsetlinewidth{0.254mm}\pgfline{\pgfpoint{34.29mm}{-24.13mm}}{\pgfpoint{6.35mm}{-24.13mm}} 
        \color[rgb]{0,0,0}\pgfsetlinewidth{0.254mm}\pgfline{\pgfpoint{50.8mm}{-10.16mm}}{\pgfpoint{36.83mm}{-24.13mm}} 
        \color[rgb]{0,0,0}\pgfsetlinewidth{0.254mm}\pgfline{\pgfpoint{50.8mm}{-10.16mm}}{\pgfpoint{64.77mm}{-24.13mm}} 
        \color[rgb]{0,0,0}\pgfsetlinewidth{0.254mm}\pgfline{\pgfpoint{64.77mm}{-24.13mm}}{\pgfpoint{36.83mm}{-24.13mm}} 
        \color[rgb]{0,0,0}\pgfsetlinewidth{0.254mm}\pgfline{\pgfpoint{20.32mm}{-21.59mm}}{\pgfpoint{6.35mm}{-24.13mm}} 
        \color[rgb]{0,0,0}\pgfsetlinewidth{0.254mm}\pgfline{\pgfpoint{20.32mm}{-21.59mm}}{\pgfpoint{34.29mm}{-24.13mm}} 
        \color[rgb]{0,0,0}\pgfsetlinewidth{0.254mm}\pgfline{\pgfpoint{20.32mm}{-21.59mm}}{\pgfpoint{20.32mm}{-10.16mm}} 
        \color[rgb]{0,0,0}\pgfsetlinewidth{0.254mm}\pgfline{\pgfpoint{20.32mm}{-17.78mm}}{\pgfpoint{6.35mm}{-24.13mm}} 
        \color[rgb]{0,0,0}\pgfsetlinewidth{0.254mm}\pgfline{\pgfpoint{20.32mm}{-13.97mm}}{\pgfpoint{6.35mm}{-24.13mm}} 
        \color[rgb]{0,0,0}\pgfsetlinewidth{0.254mm}\pgfline{\pgfpoint{20.32mm}{-17.78mm}}{\pgfpoint{34.29mm}{-24.13mm}} 
        \color[rgb]{0,0,0}\pgfsetlinewidth{0.254mm}\pgfline{\pgfpoint{20.32mm}{-13.97mm}}{\pgfpoint{34.29mm}{-24.13mm}} 
        \color[rgb]{0,0,0}\pgfsetlinewidth{0.254mm}\pgfline{\pgfpoint{50.8mm}{-10.16mm}}{\pgfpoint{48.26mm}{-17.78mm}} 
        \color[rgb]{0,0,0}\pgfsetlinewidth{0.254mm}\pgfline{\pgfpoint{50.8mm}{-10.16mm}}{\pgfpoint{53.34mm}{-17.78mm}} 
        \color[rgb]{0,0,0}\pgfsetlinewidth{0.254mm}\pgfline{\pgfpoint{48.26mm}{-17.78mm}}{\pgfpoint{53.34mm}{-17.78mm}} 
        \color[rgb]{0,0,0}\pgfsetlinewidth{0.254mm}\pgfline{\pgfpoint{53.34mm}{-17.78mm}}{\pgfpoint{50.8mm}{-20.32mm}} 
        \color[rgb]{0,0,0}\pgfsetlinewidth{0.254mm}\pgfline{\pgfpoint{50.8mm}{-20.32mm}}{\pgfpoint{48.26mm}{-17.78mm}} 
        \color[rgb]{0,0,0}\pgfsetlinewidth{0.254mm}\pgfline{\pgfpoint{48.26mm}{-17.78mm}}{\pgfpoint{36.83mm}{-24.13mm}} 
        \color[rgb]{0,0,0}\pgfsetlinewidth{0.254mm}\pgfline{\pgfpoint{50.8mm}{-20.32mm}}{\pgfpoint{36.83mm}{-24.13mm}} 
        \color[rgb]{0,0,0}\pgfsetlinewidth{0.254mm}\pgfline{\pgfpoint{50.8mm}{-20.32mm}}{\pgfpoint{64.77mm}{-24.13mm}} 
        \color[rgb]{0,0,0}\pgfsetlinewidth{0.254mm}\pgfline{\pgfpoint{53.34mm}{-17.78mm}}{\pgfpoint{64.77mm}{-24.13mm}} 
        \color[gray]{0.7}\pgfellipse[fill]{\pgfpoint{20.32mm}{-10.16mm}}{\pgfpoint{0.508mm}{0mm}}{\pgfpoint{0mm}{0.508mm}}\color[rgb]{0,0,0}\pgfsetlinewidth{0.254mm}\pgfellipse[stroke]{\pgfpoint{20.32mm}{-10.16mm}}{\pgfpoint{0.508mm}{0mm}}{\pgfpoint{0mm}{0.508mm}} 
        \color[gray]{0.7}\pgfellipse[fill]{\pgfpoint{20.32mm}{-13.97mm}}{\pgfpoint{0.508mm}{0mm}}{\pgfpoint{0mm}{0.508mm}}\color[rgb]{0,0,0}\pgfsetlinewidth{0.254mm}\pgfellipse[stroke]{\pgfpoint{20.32mm}{-13.97mm}}{\pgfpoint{0.508mm}{0mm}}{\pgfpoint{0mm}{0.508mm}} 
        \color[gray]{0.7}\pgfellipse[fill]{\pgfpoint{20.32mm}{-17.78mm}}{\pgfpoint{0.508mm}{0mm}}{\pgfpoint{0mm}{0.508mm}}\color[rgb]{0,0,0}\pgfsetlinewidth{0.254mm}\pgfellipse[stroke]{\pgfpoint{20.32mm}{-17.78mm}}{\pgfpoint{0.508mm}{0mm}}{\pgfpoint{0mm}{0.508mm}} 
        \color[gray]{0.7}\pgfellipse[fill]{\pgfpoint{20.32mm}{-21.59mm}}{\pgfpoint{0.508mm}{0mm}}{\pgfpoint{0mm}{0.508mm}}\color[rgb]{0,0,0}\pgfsetlinewidth{0.254mm}\pgfellipse[stroke]{\pgfpoint{20.32mm}{-21.59mm}}{\pgfpoint{0.508mm}{0mm}}{\pgfpoint{0mm}{0.508mm}} 
        \color[gray]{0.7}\pgfellipse[fill]{\pgfpoint{6.35mm}{-24.13mm}}{\pgfpoint{0.508mm}{0mm}}{\pgfpoint{0mm}{0.508mm}}\color[rgb]{0,0,0}\pgfsetlinewidth{0.254mm}\pgfellipse[stroke]{\pgfpoint{6.35mm}{-24.13mm}}{\pgfpoint{0.508mm}{0mm}}{\pgfpoint{0mm}{0.508mm}} 
        \color[gray]{0.7}\pgfellipse[fill]{\pgfpoint{34.29mm}{-24.13mm}}{\pgfpoint{0.508mm}{0mm}}{\pgfpoint{0mm}{0.508mm}}\color[rgb]{0,0,0}\pgfsetlinewidth{0.254mm}\pgfellipse[stroke]{\pgfpoint{34.29mm}{-24.13mm}}{\pgfpoint{0.508mm}{0mm}}{\pgfpoint{0mm}{0.508mm}} 
        \color[gray]{0.7}\pgfellipse[fill]{\pgfpoint{36.83mm}{-24.13mm}}{\pgfpoint{0.508mm}{0mm}}{\pgfpoint{0mm}{0.508mm}}\color[rgb]{0,0,0}\pgfsetlinewidth{0.254mm}\pgfellipse[stroke]{\pgfpoint{36.83mm}{-24.13mm}}{\pgfpoint{0.508mm}{0mm}}{\pgfpoint{0mm}{0.508mm}} 
        \color[gray]{0.7}\pgfellipse[fill]{\pgfpoint{50.8mm}{-20.32mm}}{\pgfpoint{0.508mm}{0mm}}{\pgfpoint{0mm}{0.508mm}}\color[rgb]{0,0,0}\pgfsetlinewidth{0.254mm}\pgfellipse[stroke]{\pgfpoint{50.8mm}{-20.32mm}}{\pgfpoint{0.508mm}{0mm}}{\pgfpoint{0mm}{0.508mm}} 
        \color[gray]{0.7}\pgfellipse[fill]{\pgfpoint{48.26mm}{-17.78mm}}{\pgfpoint{0.508mm}{0mm}}{\pgfpoint{0mm}{0.508mm}}\color[rgb]{0,0,0}\pgfsetlinewidth{0.254mm}\pgfellipse[stroke]{\pgfpoint{48.26mm}{-17.78mm}}{\pgfpoint{0.508mm}{0mm}}{\pgfpoint{0mm}{0.508mm}} 
        \color[gray]{0.7}\pgfellipse[fill]{\pgfpoint{53.34mm}{-17.78mm}}{\pgfpoint{0.508mm}{0mm}}{\pgfpoint{0mm}{0.508mm}}\color[rgb]{0,0,0}\pgfsetlinewidth{0.254mm}\pgfellipse[stroke]{\pgfpoint{53.34mm}{-17.78mm}}{\pgfpoint{0.508mm}{0mm}}{\pgfpoint{0mm}{0.508mm}} 
        \color[gray]{0.7}\pgfellipse[fill]{\pgfpoint{50.8mm}{-10.16mm}}{\pgfpoint{0.508mm}{0mm}}{\pgfpoint{0mm}{0.508mm}}\color[rgb]{0,0,0}\pgfsetlinewidth{0.254mm}\pgfellipse[stroke]{\pgfpoint{50.8mm}{-10.16mm}}{\pgfpoint{0.508mm}{0mm}}{\pgfpoint{0mm}{0.508mm}} 
        \color[gray]{0.7}\pgfellipse[fill]{\pgfpoint{64.77mm}{-24.13mm}}{\pgfpoint{0.508mm}{0mm}}{\pgfpoint{0mm}{0.508mm}}\color[rgb]{0,0,0}\pgfsetlinewidth{0.254mm}\pgfellipse[stroke]{\pgfpoint{64.77mm}{-24.13mm}}{\pgfpoint{0.508mm}{0mm}}{\pgfpoint{0mm}{0.508mm}} 
        \color[rgb]{0,0,0} 
\end{pgfpicture}    \caption{All 1-skeleta of convex simplicial polyhedra on 6 vertices.}
    \label{fig:1-skeleta-v6}
  \end{figure}

 \begin{figure}[h]
  \begin{minipage}[h]{0.5\linewidth} 
    \centering
  \begin{pgfpicture}
        \color[rgb]{0,0,0}\pgfsetlinewidth{0.254mm}\pgfline{\pgfpoint{20.32mm}{-10.16mm}}{\pgfpoint{10.16mm}{-20.32mm}}
        \color[rgb]{0,0,0}\pgfsetlinewidth{0.254mm}\pgfline{\pgfpoint{20.32mm}{-10.16mm}}{\pgfpoint{30.48mm}{-20.32mm}}
        \color[rgb]{0,0,0}\pgfsetlinewidth{0.254mm}\pgfline{\pgfpoint{30.48mm}{-20.32mm}}{\pgfpoint{10.16mm}{-20.32mm}}
        \color[rgb]{0,0,0}\pgfsetlinewidth{0.254mm}\pgfline{\pgfpoint{10.16mm}{-30.48mm}}{\pgfpoint{30.48mm}{-30.48mm}}
        \color[rgb]{0,0,0}\pgfsetlinewidth{0.254mm}\pgfline{\pgfpoint{30.48mm}{-30.48mm}}{\pgfpoint{30.48mm}{-40.64mm}}
        \color[rgb]{0,0,0}\pgfsetlinewidth{0.254mm}\pgfline{\pgfpoint{30.48mm}{-40.64mm}}{\pgfpoint{10.16mm}{-40.64mm}}
        \color[rgb]{0,0,0}\pgfsetlinewidth{0.254mm}\pgfline{\pgfpoint{10.16mm}{-40.64mm}}{\pgfpoint{10.16mm}{-30.48mm}}
        \color[rgb]{0,0,0}\pgfsetlinewidth{0.254mm}\pgfline{\pgfpoint{20.32mm}{-50.8mm}}{\pgfpoint{10.16mm}{-55.88mm}}
        \color[rgb]{0,0,0}\pgfsetlinewidth{0.254mm}\pgfline{\pgfpoint{20.32mm}{-50.8mm}}{\pgfpoint{30.48mm}{-55.88mm}}
        \color[rgb]{0,0,0}\pgfsetlinewidth{0.254mm}\pgfline{\pgfpoint{30.48mm}{-55.88mm}}{\pgfpoint{27.94mm}{-63.5mm}}
        \color[rgb]{0,0,0}\pgfsetlinewidth{0.254mm}\pgfline{\pgfpoint{27.94mm}{-63.5mm}}{\pgfpoint{12.7mm}{-63.5mm}}
        \color[rgb]{0,0,0}\pgfsetlinewidth{0.254mm}\pgfline{\pgfpoint{12.7mm}{-63.5mm}}{\pgfpoint{10.16mm}{-55.88mm}}
        \color[rgb]{0,0,0}\pgfsetlinewidth{0.254mm}\pgfline{\pgfpoint{10.16mm}{-30.48mm}}{\pgfpoint{30.48mm}{-40.64mm}}
        \color[rgb]{0,0,0}\pgfsetlinewidth{0.254mm}\pgfline{\pgfpoint{20.32mm}{-50.8mm}}{\pgfpoint{27.94mm}{-63.5mm}}
        \color[rgb]{0,0,0}\pgfsetlinewidth{0.254mm}\pgfline{\pgfpoint{12.7mm}{-63.5mm}}{\pgfpoint{20.32mm}{-50.8mm}}
        \color[rgb]{0,0,0}\pgfsetlinewidth{0.254mm}\pgfline{\pgfpoint{53.34mm}{-10.16mm}}{\pgfpoint{43.18mm}{-20.32mm}}
        \color[rgb]{0,0,0}\pgfsetlinewidth{0.254mm}\pgfline{\pgfpoint{53.34mm}{-10.16mm}}{\pgfpoint{63.5mm}{-20.32mm}}
        \color[rgb]{0,0,0}\pgfsetlinewidth{0.254mm}\pgfline{\pgfpoint{63.5mm}{-20.32mm}}{\pgfpoint{43.18mm}{-20.32mm}}
        \color[rgb]{0,0,0}\pgfsetlinewidth{0.254mm}\pgfline{\pgfpoint{43.18mm}{-30.48mm}}{\pgfpoint{63.5mm}{-30.48mm}}
        \color[rgb]{0,0,0}\pgfsetlinewidth{0.254mm}\pgfline{\pgfpoint{63.5mm}{-30.48mm}}{\pgfpoint{63.5mm}{-40.64mm}}
        \color[rgb]{0,0,0}\pgfsetlinewidth{0.254mm}\pgfline{\pgfpoint{63.5mm}{-40.64mm}}{\pgfpoint{43.18mm}{-40.64mm}}
        \color[rgb]{0,0,0}\pgfsetlinewidth{0.254mm}\pgfline{\pgfpoint{43.18mm}{-40.64mm}}{\pgfpoint{43.18mm}{-30.48mm}}
        \color[rgb]{0,0,0}\pgfsetlinewidth{0.254mm}\pgfline{\pgfpoint{53.34mm}{-50.8mm}}{\pgfpoint{43.18mm}{-55.88mm}}
        \color[rgb]{0,0,0}\pgfsetlinewidth{0.254mm}\pgfline{\pgfpoint{53.34mm}{-50.8mm}}{\pgfpoint{63.5mm}{-55.88mm}}
        \color[rgb]{0,0,0}\pgfsetlinewidth{0.254mm}\pgfline{\pgfpoint{63.5mm}{-55.88mm}}{\pgfpoint{60.96mm}{-63.5mm}}
        \color[rgb]{0,0,0}\pgfsetlinewidth{0.254mm}\pgfline{\pgfpoint{60.96mm}{-63.5mm}}{\pgfpoint{45.72mm}{-63.5mm}}
        \color[rgb]{0,0,0}\pgfsetlinewidth{0.254mm}\pgfline{\pgfpoint{45.72mm}{-63.5mm}}{\pgfpoint{43.18mm}{-55.88mm}}
     
        \color[rgb]{1,0,0}\pgfsetlinewidth{0.254mm}\pgfline{\pgfpoint{53.34mm}{-10.16mm}}{\pgfpoint{53.34mm}{-16.51mm}}
        \color[rgb]{1,0,0}\pgfsetlinewidth{0.254mm}\pgfline{\pgfpoint{53.34mm}{-16.51mm}}{\pgfpoint{43.18mm}{-20.32mm}}
        \color[rgb]{1,0,0}\pgfsetlinewidth{0.254mm}\pgfline{\pgfpoint{53.34mm}{-16.51mm}}{\pgfpoint{63.5mm}{-20.32mm}}
        \color[rgb]{1,0,0}\pgfsetlinewidth{0.254mm}\pgfline{\pgfpoint{63.5mm}{-30.48mm}}{\pgfpoint{50.8mm}{-36.83mm}}
        \color[rgb]{1,0,0}\pgfsetlinewidth{0.254mm}\pgfline{\pgfpoint{43.18mm}{-30.48mm}}{\pgfpoint{50.8mm}{-36.83mm}}
        \color[rgb]{1,0,0}\pgfsetlinewidth{0.254mm}\pgfline{\pgfpoint{50.8mm}{-36.83mm}}{\pgfpoint{43.18mm}{-40.64mm}}
        \color[rgb]{1,0,0}\pgfsetlinewidth{0.254mm}\pgfline{\pgfpoint{50.8mm}{-36.83mm}}{\pgfpoint{63.5mm}{-40.64mm}}
        \color[rgb]{1,0,0}\pgfsetlinewidth{0.254mm}\pgfline{\pgfpoint{53.34mm}{-55.88mm}}{\pgfpoint{53.34mm}{-50.8mm}}
        \color[rgb]{1,0,0}\pgfsetlinewidth{0.254mm}\pgfline{\pgfpoint{53.34mm}{-55.88mm}}{\pgfpoint{43.18mm}{-55.88mm}}
        \color[rgb]{1,0,0}\pgfsetlinewidth{0.254mm}\pgfline{\pgfpoint{53.34mm}{-55.88mm}}{\pgfpoint{63.5mm}{-55.88mm}}
        \color[rgb]{1,0,0}\pgfsetlinewidth{0.254mm}\pgfline{\pgfpoint{60.96mm}{-63.5mm}}{\pgfpoint{53.34mm}{-55.88mm}}
        \color[rgb]{1,0,0}\pgfsetlinewidth{0.254mm}\pgfline{\pgfpoint{53.34mm}{-55.88mm}}{\pgfpoint{45.72mm}{-63.5mm}}
        \color[gray]{0.7}\pgfellipse[fill]{\pgfpoint{50.8mm}{-36.83mm}}{\pgfpoint{0.508mm}{0mm}}{\pgfpoint{0mm}{0.508mm}}\color[rgb]{1,0,0}\pgfsetlinewidth{0.254mm}\pgfellipse[stroke]{\pgfpoint{50.8mm}{-36.83mm}}{\pgfpoint{0.508mm}{0mm}}{\pgfpoint{0mm}{0.508mm}}
        \color[gray]{0.7}\pgfellipse[fill]{\pgfpoint{53.34mm}{-16.51mm}}{\pgfpoint{0.508mm}{0mm}}{\pgfpoint{0mm}{0.508mm}}\color[rgb]{1,0,0}\pgfsetlinewidth{0.254mm}\pgfellipse[stroke]{\pgfpoint{53.34mm}{-16.51mm}}{\pgfpoint{0.508mm}{0mm}}{\pgfpoint{0mm}{0.508mm}}
        \color[gray]{0.7}\pgfellipse[fill]{\pgfpoint{53.34mm}{-55.88mm}}{\pgfpoint{0.508mm}{0mm}}{\pgfpoint{0mm}{0.508mm}}\color[rgb]{1,0,0}\pgfsetlinewidth{0.254mm}\pgfellipse[stroke]{\pgfpoint{53.34mm}{-55.88mm}}{\pgfpoint{0.508mm}{0mm}}{\pgfpoint{0mm}{0.508mm}}
        \color[gray]{0.7}\pgfellipse[fill]{\pgfpoint{20.32mm}{-10.16mm}}{\pgfpoint{0.508mm}{0mm}}{\pgfpoint{0mm}{0.508mm}}\color[rgb]{0,0,0}\pgfsetlinewidth{0.254mm}\pgfellipse[stroke]{\pgfpoint{20.32mm}{-10.16mm}}{\pgfpoint{0.508mm}{0mm}}{\pgfpoint{0mm}{0.508mm}}
        \color[gray]{0.7}\pgfellipse[fill]{\pgfpoint{10.16mm}{-20.32mm}}{\pgfpoint{0.508mm}{0mm}}{\pgfpoint{0mm}{0.508mm}}\color[rgb]{0,0,0}\pgfsetlinewidth{0.254mm}\pgfellipse[stroke]{\pgfpoint{10.16mm}{-20.32mm}}{\pgfpoint{0.508mm}{0mm}}{\pgfpoint{0mm}{0.508mm}}
        \color[gray]{0.7}\pgfellipse[fill]{\pgfpoint{30.48mm}{-20.32mm}}{\pgfpoint{0.508mm}{0mm}}{\pgfpoint{0mm}{0.508mm}}\color[rgb]{0,0,0}\pgfsetlinewidth{0.254mm}\pgfellipse[stroke]{\pgfpoint{30.48mm}{-20.32mm}}{\pgfpoint{0.508mm}{0mm}}{\pgfpoint{0mm}{0.508mm}}
        \color[gray]{0.7}\pgfellipse[fill]{\pgfpoint{30.48mm}{-30.48mm}}{\pgfpoint{0.508mm}{0mm}}{\pgfpoint{0mm}{0.508mm}}\color[rgb]{0,0,0}\pgfsetlinewidth{0.254mm}\pgfellipse[stroke]{\pgfpoint{30.48mm}{-30.48mm}}{\pgfpoint{0.508mm}{0mm}}{\pgfpoint{0mm}{0.508mm}}
        \color[gray]{0.7}\pgfellipse[fill]{\pgfpoint{10.16mm}{-30.48mm}}{\pgfpoint{0.508mm}{0mm}}{\pgfpoint{0mm}{0.508mm}}\color[rgb]{0,0,0}\pgfsetlinewidth{0.254mm}\pgfellipse[stroke]{\pgfpoint{10.16mm}{-30.48mm}}{\pgfpoint{0.508mm}{0mm}}{\pgfpoint{0mm}{0.508mm}}
        \color[gray]{0.7}\pgfellipse[fill]{\pgfpoint{10.16mm}{-40.64mm}}{\pgfpoint{0.508mm}{0mm}}{\pgfpoint{0mm}{0.508mm}}\color[rgb]{0,0,0}\pgfsetlinewidth{0.254mm}\pgfellipse[stroke]{\pgfpoint{10.16mm}{-40.64mm}}{\pgfpoint{0.508mm}{0mm}}{\pgfpoint{0mm}{0.508mm}}
        \color[gray]{0.7}\pgfellipse[fill]{\pgfpoint{30.48mm}{-40.64mm}}{\pgfpoint{0.508mm}{0mm}}{\pgfpoint{0mm}{0.508mm}}\color[rgb]{0,0,0}\pgfsetlinewidth{0.254mm}\pgfellipse[stroke]{\pgfpoint{30.48mm}{-40.64mm}}{\pgfpoint{0.508mm}{0mm}}{\pgfpoint{0mm}{0.508mm}}
        \color[gray]{0.7}\pgfellipse[fill]{\pgfpoint{20.32mm}{-50.8mm}}{\pgfpoint{0.508mm}{0mm}}{\pgfpoint{0mm}{0.508mm}}\color[rgb]{0,0,0}\pgfsetlinewidth{0.254mm}\pgfellipse[stroke]{\pgfpoint{20.32mm}{-50.8mm}}{\pgfpoint{0.508mm}{0mm}}{\pgfpoint{0mm}{0.508mm}}
        \color[gray]{0.7}\pgfellipse[fill]{\pgfpoint{10.16mm}{-55.88mm}}{\pgfpoint{0.508mm}{0mm}}{\pgfpoint{0mm}{0.508mm}}\color[rgb]{0,0,0}\pgfsetlinewidth{0.254mm}\pgfellipse[stroke]{\pgfpoint{10.16mm}{-55.88mm}}{\pgfpoint{0.508mm}{0mm}}{\pgfpoint{0mm}{0.508mm}}
        \color[gray]{0.7}\pgfellipse[fill]{\pgfpoint{12.7mm}{-63.5mm}}{\pgfpoint{0.508mm}{0mm}}{\pgfpoint{0mm}{0.508mm}}\color[rgb]{0,0,0}\pgfsetlinewidth{0.254mm}\pgfellipse[stroke]{\pgfpoint{12.7mm}{-63.5mm}}{\pgfpoint{0.508mm}{0mm}}{\pgfpoint{0mm}{0.508mm}}
        \color[gray]{0.7}\pgfellipse[fill]{\pgfpoint{27.94mm}{-63.5mm}}{\pgfpoint{0.508mm}{0mm}}{\pgfpoint{0mm}{0.508mm}}\color[rgb]{0,0,0}\pgfsetlinewidth{0.254mm}\pgfellipse[stroke]{\pgfpoint{27.94mm}{-63.5mm}}{\pgfpoint{0.508mm}{0mm}}{\pgfpoint{0mm}{0.508mm}}
        \color[gray]{0.7}\pgfellipse[fill]{\pgfpoint{30.48mm}{-55.88mm}}{\pgfpoint{0.508mm}{0mm}}{\pgfpoint{0mm}{0.508mm}}\color[rgb]{0,0,0}\pgfsetlinewidth{0.254mm}\pgfellipse[stroke]{\pgfpoint{30.48mm}{-55.88mm}}{\pgfpoint{0.508mm}{0mm}}{\pgfpoint{0mm}{0.508mm}}
        \color[gray]{0.7}\pgfellipse[fill]{\pgfpoint{43.18mm}{-55.88mm}}{\pgfpoint{0.508mm}{0mm}}{\pgfpoint{0mm}{0.508mm}}\color[rgb]{0,0,0}\pgfsetlinewidth{0.254mm}\pgfellipse[stroke]{\pgfpoint{43.18mm}{-55.88mm}}{\pgfpoint{0.508mm}{0mm}}{\pgfpoint{0mm}{0.508mm}}
        \color[gray]{0.7}\pgfellipse[fill]{\pgfpoint{53.34mm}{-50.8mm}}{\pgfpoint{0.508mm}{0mm}}{\pgfpoint{0mm}{0.508mm}}\color[rgb]{0,0,0}\pgfsetlinewidth{0.254mm}\pgfellipse[stroke]{\pgfpoint{53.34mm}{-50.8mm}}{\pgfpoint{0.508mm}{0mm}}{\pgfpoint{0mm}{0.508mm}}
        \color[gray]{0.7}\pgfellipse[fill]{\pgfpoint{45.72mm}{-63.5mm}}{\pgfpoint{0.508mm}{0mm}}{\pgfpoint{0mm}{0.508mm}}\color[rgb]{0,0,0}\pgfsetlinewidth{0.254mm}\pgfellipse[stroke]{\pgfpoint{45.72mm}{-63.5mm}}{\pgfpoint{0.508mm}{0mm}}{\pgfpoint{0mm}{0.508mm}}
        \color[gray]{0.7}\pgfellipse[fill]{\pgfpoint{60.96mm}{-63.5mm}}{\pgfpoint{0.508mm}{0mm}}{\pgfpoint{0mm}{0.508mm}}\color[rgb]{0,0,0}\pgfsetlinewidth{0.254mm}\pgfellipse[stroke]{\pgfpoint{60.96mm}{-63.5mm}}{\pgfpoint{0.508mm}{0mm}}{\pgfpoint{0mm}{0.508mm}}
        \color[gray]{0.7}\pgfellipse[fill]{\pgfpoint{63.5mm}{-55.88mm}}{\pgfpoint{0.508mm}{0mm}}{\pgfpoint{0mm}{0.508mm}}\color[rgb]{0,0,0}\pgfsetlinewidth{0.254mm}\pgfellipse[stroke]{\pgfpoint{63.5mm}{-55.88mm}}{\pgfpoint{0.508mm}{0mm}}{\pgfpoint{0mm}{0.508mm}}
        \color[gray]{0.7}\pgfellipse[fill]{\pgfpoint{63.5mm}{-40.64mm}}{\pgfpoint{0.508mm}{0mm}}{\pgfpoint{0mm}{0.508mm}}\color[rgb]{0,0,0}\pgfsetlinewidth{0.254mm}\pgfellipse[stroke]{\pgfpoint{63.5mm}{-40.64mm}}{\pgfpoint{0.508mm}{0mm}}{\pgfpoint{0mm}{0.508mm}}
        \color[gray]{0.7}\pgfellipse[fill]{\pgfpoint{43.18mm}{-40.64mm}}{\pgfpoint{0.508mm}{0mm}}{\pgfpoint{0mm}{0.508mm}}\color[rgb]{0,0,0}\pgfsetlinewidth{0.254mm}\pgfellipse[stroke]{\pgfpoint{43.18mm}{-40.64mm}}{\pgfpoint{0.508mm}{0mm}}{\pgfpoint{0mm}{0.508mm}}
        \color[gray]{0.7}\pgfellipse[fill]{\pgfpoint{43.18mm}{-30.48mm}}{\pgfpoint{0.508mm}{0mm}}{\pgfpoint{0mm}{0.508mm}}\color[rgb]{0,0,0}\pgfsetlinewidth{0.254mm}\pgfellipse[stroke]{\pgfpoint{43.18mm}{-30.48mm}}{\pgfpoint{0.508mm}{0mm}}{\pgfpoint{0mm}{0.508mm}}
        \color[gray]{0.7}\pgfellipse[fill]{\pgfpoint{63.5mm}{-30.48mm}}{\pgfpoint{0.508mm}{0mm}}{\pgfpoint{0mm}{0.508mm}}\color[rgb]{0,0,0}\pgfsetlinewidth{0.254mm}\pgfellipse[stroke]{\pgfpoint{63.5mm}{-30.48mm}}{\pgfpoint{0.508mm}{0mm}}{\pgfpoint{0mm}{0.508mm}}
        \color[gray]{0.7}\pgfellipse[fill]{\pgfpoint{63.5mm}{-20.32mm}}{\pgfpoint{0.508mm}{0mm}}{\pgfpoint{0mm}{0.508mm}}\color[rgb]{0,0,0}\pgfsetlinewidth{0.254mm}\pgfellipse[stroke]{\pgfpoint{63.5mm}{-20.32mm}}{\pgfpoint{0.508mm}{0mm}}{\pgfpoint{0mm}{0.508mm}}
        \color[gray]{0.7}\pgfellipse[fill]{\pgfpoint{53.34mm}{-10.16mm}}{\pgfpoint{0.508mm}{0mm}}{\pgfpoint{0mm}{0.508mm}}\color[rgb]{0,0,0}\pgfsetlinewidth{0.254mm}\pgfellipse[stroke]{\pgfpoint{53.34mm}{-10.16mm}}{\pgfpoint{0.508mm}{0mm}}{\pgfpoint{0mm}{0.508mm}}
        \color[gray]{0.7}\pgfellipse[fill]{\pgfpoint{43.18mm}{-20.32mm}}{\pgfpoint{0.508mm}{0mm}}{\pgfpoint{0mm}{0.508mm}}\color[rgb]{0,0,0}\pgfsetlinewidth{0.254mm}\pgfellipse[stroke]{\pgfpoint{43.18mm}{-20.32mm}}{\pgfpoint{0.508mm}{0mm}}{\pgfpoint{0mm}{0.508mm}}
        \color[rgb]{0,0,0}\pgfsetlinewidth{0.254mm}\pgfline{\pgfpoint{33.02mm}{-15.24mm}}{\pgfpoint{40.64mm}{-15.24mm}}
        \color[rgb]{0,0,0}\pgfsetlinewidth{0.254mm}\pgfline{\pgfpoint{39.37mm}{-13.97mm}}{\pgfpoint{40.64mm}{-15.24mm}}
        \color[rgb]{0,0,0}\pgfsetlinewidth{0.254mm}\pgfline{\pgfpoint{39.37mm}{-16.51mm}}{\pgfpoint{40.64mm}{-15.24mm}}
        \color[rgb]{0,0,0}\pgfsetlinewidth{0.254mm}\pgfline{\pgfpoint{33.02mm}{-35.56mm}}{\pgfpoint{40.64mm}{-35.56mm}}
        \color[rgb]{0,0,0}\pgfsetlinewidth{0.254mm}\pgfline{\pgfpoint{39.37mm}{-34.29mm}}{\pgfpoint{40.64mm}{-35.56mm}}
        \color[rgb]{0,0,0}\pgfsetlinewidth{0.254mm}\pgfline{\pgfpoint{39.37mm}{-36.83mm}}{\pgfpoint{40.64mm}{-35.56mm}}
        \color[rgb]{0,0,0}\pgfsetlinewidth{0.254mm}\pgfline{\pgfpoint{33.02mm}{-58.42mm}}{\pgfpoint{40.64mm}{-58.42mm}}
        \color[rgb]{0,0,0}\pgfsetlinewidth{0.254mm}\pgfline{\pgfpoint{40.64mm}{-58.42mm}}{\pgfpoint{39.37mm}{-57.15mm}}
        \color[rgb]{0,0,0}\pgfsetlinewidth{0.254mm}\pgfline{\pgfpoint{39.37mm}{-59.69mm}}{\pgfpoint{40.64mm}{-58.42mm}}
        \color[rgb]{0,0,0}\pgfputat{\pgfpoint{33.02mm}{-12.7mm}}{\pgfbox[left,center]{\shortstack{$H_1$}}}
        \color[rgb]{0,0,0}\pgfputat{\pgfpoint{33.02mm}{-33.02mm}}{\pgfbox[left,center]{\shortstack{$H_2$}}}
        \color[rgb]{0,0,0}\pgfputat{\pgfpoint{33.02mm}{-55.88mm}}{\pgfbox[left,center]{\shortstack{$H_3$}}}
   \color[rgb]{0,0,0}\pgfsetlinewidth{0.254mm}\pgfsetdash{{5pt}{5pt}{5pt}{5pt}}{0mm}\pgfline{\pgfpoint{43.18mm}{-30.48mm}}{\pgfpoint{63.5mm}{-40.64mm}}
        \color[rgb]{0,0,0}\pgfsetlinewidth{0.254mm}\pgfsetdash{{5pt}{5pt}{5pt}{5pt}}{0mm}\pgfline{\pgfpoint{53.34mm}{-50.8mm}}{\pgfpoint{60.96mm}{-63.5mm}}
        \color[rgb]{0,0,0}\pgfsetlinewidth{0.254mm}\pgfsetdash{{5pt}{5pt}{5pt}{5pt}}{0mm}\pgfline{\pgfpoint{45.72mm}{-63.5mm}}{\pgfpoint{53.34mm}{-50.8mm}}
        \color[rgb]{0,0,0} 
\end{pgfpicture}

  \caption{The spatial Henneberg steps. 
    These are analyzed in Section~\protect\ref{Sgenerate3d}.}
  \label{hennsteps} 
  \end{minipage}
  \hspace{0.1cm} 
  \begin{minipage}[h]{0.5\linewidth}
    \centering
      \includegraphics[scale=0.37]{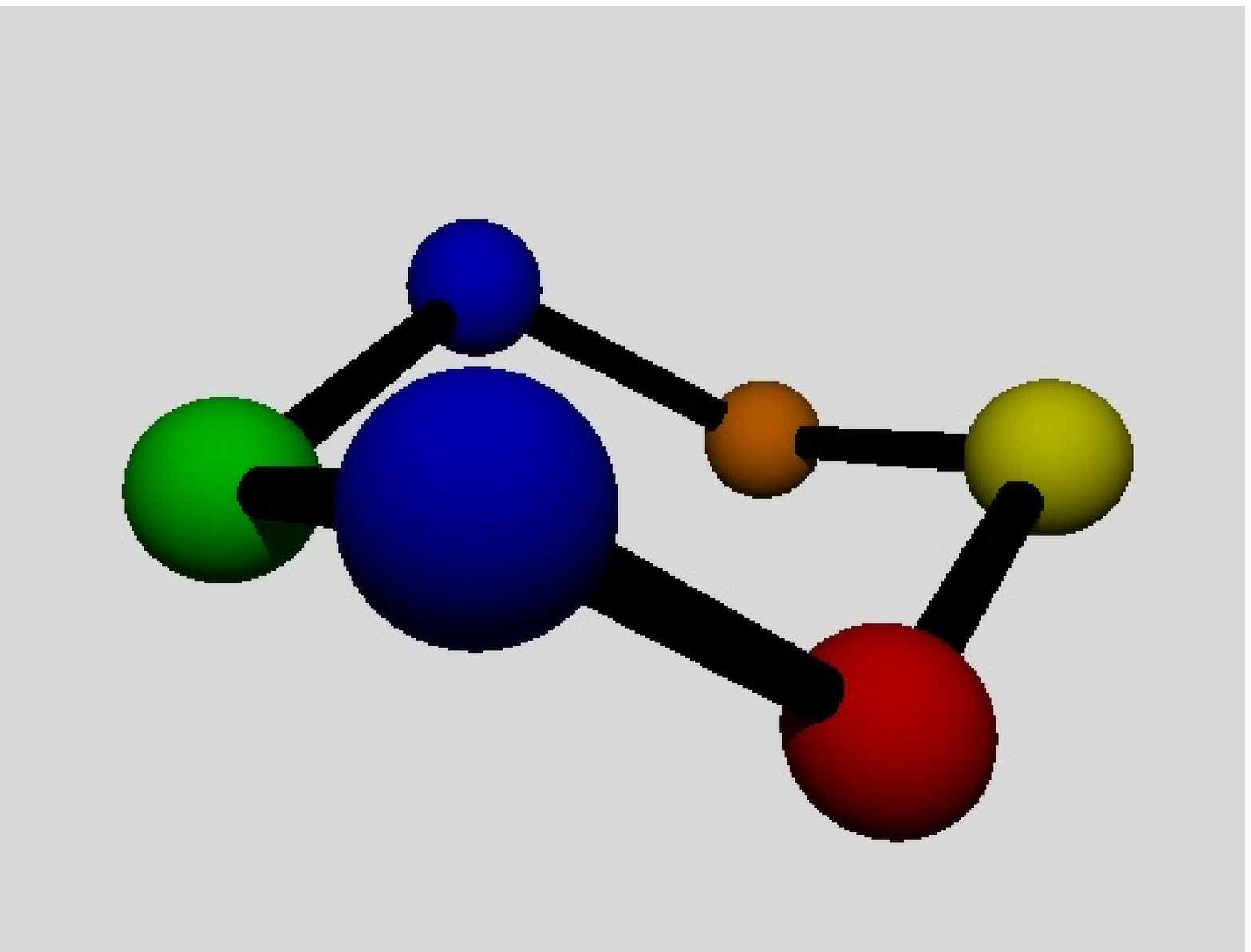}
  \label{cyclohexane configuration}
    \caption{One of the spatial (chair) configurations of the cyclohexane molecule.}
        \label{fig:cyclo_chair}

      \end{minipage}
\end{figure}

\begin{figure}[t]
   \begin{minipage}[b]{0.5\linewidth} 
    \centering 
     \includegraphics{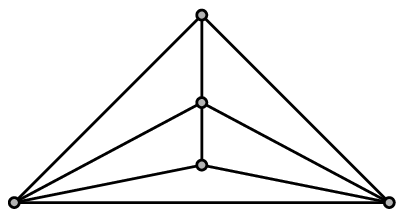}
     \caption{The only 1-skeleton of a simplicial polytope on 5 vertices.}
     \label{spatial5}
    \end{minipage}
  \hspace{0.1cm} 
  \begin{minipage}[b]{0.5\linewidth}
    \centering \includegraphics[scale=0.8]{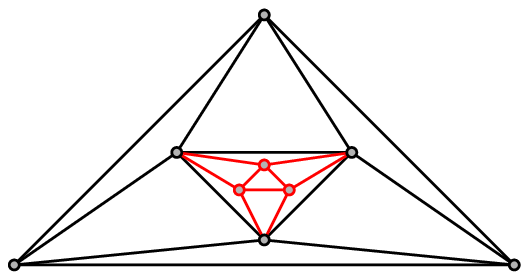}
    \caption{A Cyclohexane caterpillar with 2 copies.}
    \label{fig:cyclohexane}
  \end{minipage}
\end{figure}

\end{document}